\documentclass[12pt]{amsart}
\usepackage{a4wide,hyperref} 

\newtheorem{theorem}{Theorem}
\newtheorem{lemma}[theorem]{Lemma}
\newtheorem{corollary}[theorem]{Corollary}
\newtheorem{proposition}[theorem]{Proposition}
\newtheorem{definition}[theorem]{Definition}
\newtheorem{remark}[theorem]{Remark}
\newtheorem{example}[theorem]{Example}
\newtheorem{notation}[theorem]{Notation}

\title{Balance properties of Arnoux-Rauzy words}

\author{Val\'erie Berth\'e}
\address{LIAFA, CNRS, Universit\'e Paris Diderot -- Paris 7, Case 7014, 75205 Paris Cedex 13, France}
\email{berthe@liafa.univ-paris-diderot.fr \vspace{-.4ex}}

\author{Julien Cassaigne}
\address{IML, CNRS, Campus de Luminy, Case 907, 13288 Marseille Cedex 9, France}
\email{cassaigne@iml.univ-mrs.fr \vspace{-.4ex}}

\author{Wolfgang Steiner}
\address{LIAFA, CNRS, Universit\'e Paris Diderot -- Paris 7, Case 7014, 75205 Paris Cedex 13, France}
\email{steiner@liafa.univ-paris-diderot.fr}

\begin{document}
\begin{abstract}
The paper deals with balances and imbalances in Arnoux-Rauzy words. 
We provide sufficient conditions for $C$-balancedness, but our results indicate that even a characterization of $2$-balanced Arnoux-Rauzy words on a $3$-letter alphabet is not immediate. 
\end{abstract}
\maketitle

\section{Introduction}
Among infinite words with low factor complexity, Arnoux-Rauzy words, introduced in~\cite{AR91}, play a prominent role. 
Recall that factor complexity consists in counting the number of factors of a given length, and by low complexity we mean here that its growth is at most linear. 
A~seminal family of low factor complexity words is provided by the Sturmian words (see for instance~\cite{Lot02,Pyt02,BLRS09}), and by  their finite counterpart, namely Christoffel words~\cite{BR06,KR07,BLR08,PR11,BLRS09}. 
Sturmian words are defined in purely combinatorial terms, but their associated symbolic dynamical systems also have a very natural geometric description: they are natural codings of rotations of the circle.
Arnoux-Rauzy sequences were introduced in~\cite{AR91} in an attempt to find a class of words generalizing properties of Sturmian words.
Again, Arnoux-Rauzy words are defined in purely combinatorial terms (see Definition~\ref{def:ar} below), and they admit geometric representations as natural codings of $6$-interval exchange transformations on the circle. 
We quote~\cite{BLR08,BFPRR12} as illustrations of common behavior shared by Sturmian and Arnoux-Rauzy words concerning the relations between bifix codes and subgroups of the free group.

A~wide array of literature is devoted to the study of the combinatorial, ergodic and geometric properties of Arnoux-Rauzy words, which belong to the family of episturmian words~\cite{GJ09}. 
The language of Arnoux-Rauzy words can be generated by iterating a finite set of substitutions; this property, which is called $S$-adic, will be the main viewpoint of our study; see Theorem~\ref{theo:ar} below. 
The $S$-adic generation is governed by a multidimensional continued fraction algorithm of simultaneous approximation described and studied for instance in~\cite{AR91,CMR99,CC06,CFM08}.  

Despite the fact that they were introduced as generalizations of Sturmian words, Arnoux-Rauzy words display a much more complex behavior, which is not yet fully understood; this is highlighted in~\cite{CFM08}. 
Some of the properties shared by purely substitutive Arnoux-Rauzy words   sustain the similarity with Sturmian words. 
For instance, Arnoux-Rauzy substitutions are known to be Pisot~\cite{AI01} and thus to generate finitely balanced Arnoux-Rauzy words.
Here, finitely balanced means that they are $C$-balanced for some $C > 0$; see Definition~\ref{d:balanced} below.
This is the case of the Tribonacci word, which is the Arnoux-Rauzy word fixed by the substitution $1\mapsto 12,\ 2 \mapsto 13,\ 3 \mapsto 1$; see for instance~\cite{RSZ10}. 
Recall that Sturmian words are $1$-balanced. 
Furthermore, purely substitutive Arnoux-Rauzy words are natural codings of toral translations; see~\cite{BJS12} for a proof.
More generally, all Arnoux-Rauzy words were expected to be natural codings of toral translations, a property that implies finite balance; see the discussions in~\cite{CFZ00,BFZ05,CFM08} for more details. 
The first striking fact contradicting this idea came from~\cite{CFZ00}, where examples of Arnoux-Rauzy sequences that are not finitely balanced were constructed.

The aim of the present  paper is to try to understand balances and imbalances in Arnoux-Rauzy words, continuing the study performed in~\cite{CFZ00}. 
We provide sufficient conditions that guarantee finite balance. 
In particular, we prove that bounded partial quotients in the $S$-adic expansion imply finite balance (Theorem~\ref{t:1}), but we also exhibit examples of $2$-balanced Arnoux-Rauzy words with unbounded partial quotients (Corollary~\ref{cor:2}). 
Our results show that a characterization of $2$-balanced Arnoux-Rauzy words on a $3$-letter alphabet is not at hand. 
Indeed, the occurrence of certain  patterns as prefixes of the $S$-adic expansion prevent $2$-balance (see Theorem~\ref{t:3}), but the image by an Arnoux-Rauzy substitution of an Arnoux-Rauzy word that is not $2$-balanced can be a $2$-balanced sequence. 

Note that balance properties of Arnoux-Rauzy words, and more generally of episturmian words, have also been investigated in view of Fraenkel's conjecture; see for instance~\cite{PV07} and see also the survey on balanced words~\cite{Vui03} for more on this conjecture.
Moreover, if superimposition properties of Christoffel words, such as considered in~\cite{PR11}, are now well understood, much remains to be done in this direction for alphabets with more than $2$~letters and for epiChristoffel words such as developped in~\cite{Paq09}. 

The present paper is organized as follows. 
After recalling the required definitions in Section~\ref{sec:def},  we state the main results of the paper in Section~\ref{sec:res}. 
Proofs  are given in Section~\ref{sec:proofs}, with Section~\ref{subsec:balance} devoted to balance results, while Section~\ref{subsec:unbalance} shows how to create imbalances.

\section{Definitions}\label{sec:def}

Let $\mathcal{A} = \{1,2,\ldots,d\}$ be a finite alphabet.  
A~\emph{substitution} $\sigma$ over the alphabet~$\mathcal{A}$ is an endomorphism of the free monoid~$\mathcal{A}^*$. 
For any word~$w$ in the free monoid~$\mathcal{A}^*$ (endowed with the concatenation as operation), $|w|$~denotes the length of~$w$, and $|w|_j$ stands for the number of occurrences of the letter~$j$ in the word~$w$.
A~\emph{factor} of a (finite or infinite) word~$\omega$ is defined as the concatenation of consecutive letters occurring in~$\omega$. 
A~factor~$w$ of~$\omega$ is said to be right special (resp.\ left special) if there exist at least two distinct letters $a,b$ of the alphabet~$\mathcal{A}$ such that $wa$ and~$wb$  (resp.\ $aw$ and~$bw$) are factors of~$\omega$. 

\begin{definition}[Arnoux-Rauzy words]  \label{def:ar} 
Let $\mathcal{A} = \{1,2,\ldots,d\}$. 
An infinite word $\omega \in \mathcal{A}^\mathbb{N}$ is an \emph{Arnoux-Rauzy word} if all its factors occur infinitely often, and if for all~$n$ it has $(d-1)n+1$ factors of length~$n$, with exactly one left special and one right special factor of length~$n$.
\end{definition}

One checks that Arnoux-Rauzy words are uniformly recurrent, that is, all factors occur with bounded gaps. 
The symbolic dynamical system $(X_{\omega},T)$ associated with a given Arnoux-Rauzy word~$\omega$ is defined as the set~$X_{\omega}$ of infinite words that have the same language as~$\omega$, endowed with the shift map~$T$ that sends an infinite word onto the infinite word from which the first letter has been removed.

In order to work with Arnoux-Rauzy words, we will use their representations as $S$-adic words according to the terminology of~\cite{Dur00,Pyt02}.
We first define the set~$\mathcal{S}$ of Arnoux-Rauzy substitutions as $\mathcal{S} = \{\sigma_i:\, i \in \mathcal{A}\}$, with
$$
\sigma_i:\ i \mapsto i,\ j \mapsto ji\ \mbox{for}\ j \in \mathcal{A} \setminus \{i\}\,.
$$

\begin{theorem}[\cite{AR91}] \label{theo:ar}
An infinite word $\omega \in \mathcal{A}^\mathbb{N}$ is an Arnoux-Rauzy word if and only if its set of factors coincides with the set of factors of a sequence of the form 
$$
\lim_{n\to\infty} \sigma_{i_0} \sigma_{i_1} \cdots \sigma_{i_n} (1), 
$$
where the sequence $(i_n)_{n\geq 0} \in \mathcal{A}^\mathbb{N}$ is such that every letter in~$\mathcal{A}$ occurs infinitely often in~$(i_n)_{n\ge0}$. 
Furthermore, such a sequence $(i_n)_{n\geq 0} \in \mathcal{A}^\mathbb{N}$ is uniquely defined for a given~$\omega$.

For any given Arnoux-Rauzy word, the sequence $(i_n)_{n \geq 0}$ is called the $\mathcal{S}$-\emph{directive word} of~$\omega$. 
All the Arnoux-Rauzy words that belong to the dynamical system $(X_{\omega},T)$ have the same $\mathcal{S}$-directive word.
\end{theorem}

\begin{remark}
If one takes the set $\mathcal{S}' = \{\sigma'_i:\, i \in \mathcal{A}\}$, with
$$
\sigma'_i:\ i \mapsto i,\ j \mapsto ij\ \mbox{for}\ j \in \mathcal{A} \setminus \{i\}
$$
then Arnoux-Rauzy words of the form  
$$
\lim_{n\to\infty} \sigma'_{i_0} \sigma'_{i_1} \cdots \sigma'_{i_n} (1)
$$
are called standard or characteristic Arnoux-Rauzy words. 
We  choose to work here with the set~$\mathcal{S}$ that proves to be more convenient for handling balance results.
\end{remark}

The following definition is inspired by the fact that Sturmian words also admit $S$-adic representations governed by the usual continued fraction algorithm; 
see~\cite{Lot02,Pyt02} for details.

\begin{definition}[Weak and strong partial quotients]
Let $\omega$ be an  Arnoux-Rauzy word with $\mathcal{S}$-directive sequence $(i_m)_{m \ge 0}$.
Write 
$$
i_0 i_1 i_2 \cdots = j_0^{k_0} j_1^{k_1} j_2^{k_2} \cdots,
$$
where $j_n \in \mathcal{A}$, $k_n \geq 1$, and $ j_n \neq j_{n +1}$ for all $n \ge 0$.
The powers~$k_n$ are called \emph{weak partial quotients}.

Let $(n_{\ell})_{\ell \ge 0}$ be the increasing sequence of integers satisfying $n_0 = 0$,
$$
\{i_{n_\ell}, i_{n_\ell+1}, \ldots,  i_{n_{\ell+1}}\} = \mathcal{A} \quad \mbox{and} \quad \{i_{n_\ell}, i_{n_\ell+1}, \ldots,  i_{n_{\ell+1}-1}\} \neq \mathcal{A}
$$
for all $\ell \geq 0$.
The quantity $(n_{\ell+1} - n_\ell)$ is called \emph{strong partial quotient}.
\end{definition}

\begin{notation}
Let $\omega$ be an Arnoux-Rauzy word with $\mathcal{S}$-directive sequence $(i_m)_{m \geq 0}$.
For every integer $m \ge 0$, we define~$\omega^{(m)}$ as the following Arnoux-Rauzy word with $\mathcal{S}$-directive sequence $(i_n)_{ n\ge m}$:
$$
\omega^{(m)} = \lim_{n\to\infty} \sigma_{i_m} \sigma_{i_{m+1}} \cdots \sigma_{i_n} (1).
$$
\end{notation}

Lastly, we introduce the notion of finite balance.

\begin{definition}[$C$-balance] \label{d:balanced}
Let $C \in \mathbb{N}$. 
A~pair of words $u, v \in \mathcal{A}^*$ with $|u| = |v|$ is \emph{$C$-balanced} if 
$$
-C \le |u|_j - |v|_j \le C \quad \mbox{for all}\ j \in \mathcal{A}\,.
$$
A~word $\omega \in \mathcal{A}^\mathbb{N}$ is \emph{$C$-balanced} if all pairs of factors $u, v$ of $\omega$ with $|u| = |v|$ are $C$-balanced.
A~word $\omega \in \mathcal{A}^\mathbb{N}$ is said to be \emph{finitely balanced} if there exists $C \in \mathbb{N}$ such that it is $C$-balanced.
\end{definition}

\section{Results}\label{sec:res}

We now restrict ourselves to $3$-letter Arnoux-Rauzy words (i.e., $d=3$)  for technical combinatorial reasons.

Note that Sturmian sequences are all $1$-balanced (and thus finitely balanced) words. 
This provides even a characterization of Sturmian words: Sturmian words are exactly the binary $1$-balanced words that are not eventually periodic. 
Three-letter words that are $1$-balanced are completely characterized in~\cite{Hub00}; in particular, Arnoux-Rauzy words cannot be $1$-balanced, see also~\cite{PV07}. 
The so-called Tribonacci word, which is the Arnoux-Rauzy word fixed by the substitution $1\mapsto 12,\ 2 \mapsto 13,\ 3 \mapsto 1$, with $\mathcal{S}$-directive sequence $123123\cdots$, has been proved to be $2$-balanced in~\cite{RSZ10}. 
Nevertheless, Arnoux-Rauzy words are now known to be in general far from being $2$-balanced.
They need not even be finitely balanced~\cite{CFZ00,CFM08}. 

A~natural question is to understand whether the fact that partial quotients are bounded can guarantee (and even characterize) finite balancedness.
We have seen that there are two types of partial quotients.
The following theorems provide sufficient conditions for $({2h\!+\!1})$-balance and $2$-balance respectively.

\begin{theorem} \label{t:1}
Let $\omega$ be an Arnoux-Rauzy word with $d = 3$ and $\mathcal{S}$-directive sequence $(i_m)_{m \geq 0}$.
If the weak partial quotients are bounded by~$h$, i.e., if we do not have $i_m = i_{m+1} = \cdots = i_{m+h}$ for any $m \ge 0$, then $\omega$ is $({2h\!+\!1})$-balanced.
\end{theorem}
 
\begin{theorem} \label{t:2}
Let $\omega$ be an Arnoux-Rauzy word with $d = 3$ and $\mathcal{S}$-directive sequence $(i_m)_{m \geq 0}$.
If $i_{m+2} \not\in \{i_m, i_{m+1}\}$ 
\begin{enumerate}
\item
for all $m \ge 1$ such that $i_{m-1} = i_m \ne i_{m+1}$  and
\item  
for all $m \ge 2$ such that $i_{m-2} = i_m \ne i_{m+1} = i_{m-1}$,
\end{enumerate}
then $\omega$ is $2$-balanced.
\end{theorem}

\begin{remark} \label{r:sft}
The conditions of Theorem~\ref{t:2} on the $\mathcal{S}$-directive sequence can also be stated in terms of a subshift of finite type: 
Let $X$ be the set of words $\{1121, 1122, 12121, 12122\}$ together with all the words that are obtained from one of these four words by a permutation of the letters~$1$, $2$, and~$3$. 
If $(i_m)_{m \geq 0}$   contains no factor in~$X$, then $\omega$ is $2$-balanced.
\end{remark}

For the proofs of Theorems~\ref{t:1} and~\ref{t:2}, which will be given in Section~\ref{subsec:balance}, we use methods that are typical for the $S$-adic framework: We work not only with the word~$\omega$ but simultaneously with all the words~$\omega^{(m)}$, by `desubstituting' the word~$\omega^{(m)}$ with respect to~$\sigma_{i_m}$.

Note that Theorem~\ref{t:2} applies in particular to the Tribonacci word; a generalization of its $\mathcal{S}$-directive sequence $123123\cdots$ gives the following balancedness result.

\begin{corollary}\label{cor:2}
There exist $2$-balanced Arnoux-Rauzy words~$\omega$ with unbounded weak (and thus unbounded strong) partial quotients.
\end{corollary}

\begin{proof}
If  the $\mathcal{S}$-directive sequence $(i_m)_{m \geq 0}$ satisfies $i_0 i_1 \cdots = 1^{s_0} 2 3 1^{s_1} 2 3 \cdots$ for any sequence $(s_n)_{n\ge 0}$ with $s_n \ge 1$, then the assumptions of Theorem~\ref{t:2} are satisfied. 
In other words, any Arnoux-Rauzy word with such an $\mathcal{S}$-directive sequence is $2$-balanced.  
\end{proof}

The following result, which is proved in Section~\ref{subsec:unbalance}, can be considered as a partial converse of Theorem~\ref{t:2}. 
It shows in particular that no element of the set~$X$ defined in Remark~\ref{r:sft} can be omitted: If~a word in~$X$ is a prefix of an $\mathcal{S}$-directive sequence, then the corresponding Arnoux-Rauzy word is not $2$-balanced.
(By symmetry of~$\mathcal{S}$, the statement of Theorem~\ref{t:3} remains of course true if we apply a permutation of the alphabet~$\mathcal{A}$ to $\{1,3\}^*\, 1\, 2^*\, 1\, 2\, \{1,2\}$.)
This extends a result of~\cite{CFZ00}, where it was shown that $\mathcal{S}$-directive sequences starting with $112213$ give Arnoux-Rauzy words that are not $2$-balanced. 
For arbitrary~$C$, prefixes preventing $C$-balancedness are also given in~\cite{CFZ00}.

\begin{theorem} \label{t:3}
Let $\omega$ be an Arnoux-Rauzy word with $d = 3$ and $\mathcal{S}$-directive sequence   $(i_m)_{m \geq 0}$
 starting with a word in $\{1,3\}^*\, 1\, 2^*\, 1\, 2\, \{1,2\}$.
Then $\omega$ is not $2$-balanced.
\end{theorem}

Theorem~\ref{t:3} also shows that Theorem~\ref{t:1} is optimal for $h = 1$: 
There exist Arnoux-Rauzy words with all weak partial quotients equal to~$1$ that are not $2$-balanced, e.g.\ those where the $\mathcal{S}$-directive sequence  starts with~$12121$.

While it is possible to exclude $C$-balancedness by looking at a finite prefix of the $\mathcal{S}$-directive sequence, it was shown in~\cite{CFZ00} that one cannot guarantee $C$-balancedness in this way: For any given imbalance~$C$, for any word $w \in \{1,2,3\}^*$, for any infinite $\mathcal{S}$-directive sequence $(i_m)_{m \geq 0}$, there exists a word $w' \in \{1,2,3\}^*$ such that the $\mathcal{S}$-directive sequence $w w' i_0 i_1 \cdots$ gives an Arnoux-Rauzy word that is not $C$-balanced.

Proposition~\ref{prop:2} below gives another evidence that the situation is quite contrasted: It is possible that $\omega$ is $2$-balanced, $\omega^{(1)}$~is not $2$-balanced, and $\omega^{(m)}$ is again $2$-balanced for all $m \geq 2$.
Thus, the image by an Arnoux-Rauzy substitution of an Arnoux-Rauzy word that is not $2$-balanced can be a $2$-balanced sequence. 
In particular, we cannot replace $\{1,3\}^*$ in Theorem~\ref{t:3} by the full set of words $\{1,2,3\}^*$.

\begin{proposition} \label{prop:2}
Let $\omega$ be an Arnoux-Rauzy word with $d = 3$ and $\mathcal{S}$-directive sequence   $(i_m)_{m \geq 0}$
 starting with $211213$ such that the conditions of Theorem~\ref{t:2} are violated only for $m = 2$.
Then $\omega$ is $2$-balanced (and $\omega^{(m)}$ is $2$-balanced for all $m \geq 2$), but $\omega^{(1)}$ is not $2$-balanced.
\end{proposition}

\begin{remark}
The conditions of Proposition~\ref{prop:2} can be weakened.
Indeed, we will use in the proof only that   $(i_m)_{m \geq 0}$
 starts with $211213$ and that $\omega^{(m)}$ is $2$-balanced for all $m \in \{2,3,4,6\}$ to show that $\omega$ is $2$-balanced (and $\omega^{(1)}$ is not $2$-balanced).
\end{remark}

Summing up, a~complete description of  $2$-balanced Arnoux-Rauzy words on  $3$~letters  in terms of their $\mathcal{S}$-directive sequences seems to be difficult. 
Note that our proofs strongly rely on the  fact that we work with $3$-letter alphabets. The case of  a larger   alphabet   is  certainly     even more
   intricate.

\section{Proofs} \label{sec:proofs}

\subsection{Balances}\label{subsec:balance}

In order to relate balance properties of~$\omega^{(m+1)}$ to the ones of~$\omega^{(m)}$, we first need the following `desubstitution' property.
For any factor~$u$ of~$\omega^{(m)}$, $m \ge 0$, one easily checks that there exists a unique factor of~$\omega^{(m+1)}$, denoted by~$u'$, that satisfies 
\begin{equation} 
i_m^\delta\, \sigma_{i_m}(u') = u\, i_m^\varepsilon\,, \label{e:uu'}
\end{equation}
where
$$
\delta = \left\{\begin{array}{cl}1 & \mbox{if $u$ starts with $i_m$,}\\[.5ex] 0 & \mbox{otherwise,}\end{array}\right. \quad \varepsilon = \left\{\begin{array}{cl}0 & \mbox{if $u$ ends with $i_m$,}\\[.5ex] 1 & \mbox{otherwise.}\end{array}\right.
$$
Note that the definition of $u'$ depends on~$m$. 
(More precisely, it depends only on~$i_m$.) 
For simplicity, we omit this dependence.
\begin{example} \label{ex:1}
Let $u=12131212$.  Assume that  $i_m=1$. One has $\delta=1$,  $\varepsilon=1$, and  $u'=2322$, with 
$1 \sigma_1(2322)=u1$.
\end{example}
We thus introduce the following notation.

\begin{notation} \label{d:u'}
For a factor $u$ of~$\omega^{(m)}$, $m \ge 0$, let $u'$ be the unique factor of $\omega^{(m+1)}$ satisfying~(\ref{e:uu'}).
Recursively, let $u^{(n)}$ be the factor of $\omega^{(m+n)}$ given by $u^{(0)} = u$ and $u^{(n+1)} = (u^{(n)})'$, $n \ge 0$. 
\end{notation}

\begin{example} \label{ex:2}
We continue Example \ref{ex:1} with 
 $u=12131212$.  We assume that  $i_m=1$, $i_{m+1}=i_{m+2}=2$.
One has $u^{(0)}=u$, $u^{(1)}=u'=2322$, $u^{(2)}=32$, and   $u^{(3)}=3$.
Indeed,  $1 \sigma_1(u^{(1)})=u^{(0)}1$,  $2 \sigma_2(u^{(2)})=u^{(1)}$,   and $ \sigma_2(u^{(3)})=u^{(2)}$.

\end{example}

The following lemma  provides information on the relations between the length and the number of occurrences of letters in $u$ and~$u'$.

\begin{lemma} \label{l:preimage}
Let $u$ be a factor of~$\omega^{(m)}$, $m \ge 0$. 
Then there exists $\delta_u \in \{-1,0,1\}$ such that
$$
\begin{array}{ll}
& |u'| = |u|_{i_m} + \delta_u,\\
&|u'|_j = \left\{\begin{array}{ll}|u|_j & \mbox{if}\ j \in \mathcal{A} \setminus \{i_m\},\\[.5ex] 2\, |u|_{i_m} - |u| + \delta_u & \mbox{if}\ j = i_m.\end{array}\right.
\end{array}
$$
Furthermore, $|u'| < |u|$ if $|u| \ge 2$, and $|u'| \in \{0, 1\}$ if $|u| \in \{0, 1\}$.  
\end{lemma}

\begin{example} \label{ex:3}
We continue Example \ref{ex:1} with  $u=12131212$, $i_m=1$ and $u'=2322$.
One has $|u'|=4= |u|_1$, with   $\delta_u=0$. Furthermore, $|u'|_1=0=2|u|_1 -|u|$,  $|u'|_2=|u|_2$,  $|u'|_3=|u|_3$.
\end{example}
\begin{proof}
By the definition of~$u'$, we have $|u'|_j = |u|_j$ for all $j \in \mathcal{A} \setminus \{i_m\}$, and $|u'| = |u|_{i_m} + \delta_u$ with $\delta_u = \varepsilon - \delta$.
It follows that $|u'|_{i_m} = 2\, |u|_{i_m} - |u| + \delta_u$. 
This ends the proof of the first part of the statement.

If $|u| = 0$, then we have $u' = i_m$. 
If $|u| = |u|_{i_m} \ge 1$, then we have $\delta_u = -1$, thus $|u'| < |u|$. 
Otherwise, if $u$ does not only contain the letter~$i_m$, we have $|u'| \le |u|$, with $|u'| = |u|$ if and only if $|u|_{i_m} = |u| - 1$, $\delta_u = 1$.
Since $\delta_u = 1$ means that $u$ does not start or end with~$i_m$, $|u|_{i_m} = |u| - 1$ implies that $|u| = 1$. 
\end{proof}

The following two lemmas allow us to easily exhibit a pair of factors of the same length displaying some imbalances when sufficiently large imbalances occur for factors that are not necessarily of the same length.

\begin{lemma} \label{l:unbalanced1}
Let $u, v \in \mathcal{A}^*$ and $j \in \mathcal{A}$ such that
$$
|u|_j - |v|_j > C + \max\big(0, |u|-|v|\big)\,.
$$
Then there exist factors $\hat{u}$ of~$u$ and $\hat{v}$ of~$v$ with $|\hat{u}| = |\hat{v}|= \min(|u|, |v|)$ and $|\hat{u}|_j - |\hat{v}|_j > C$.
\end{lemma}

\begin{proof}
We first assume $|u| \le |v|$,  hence $ |u|_j - |v|_j > C$.  
Let $\hat{u} = u$ and take for~$\hat{v}$ any factor of~$v$ with the same length as~$u$ (hence $|\hat{v}| = |\hat{u}|$). 
Then $|\hat{u}|_j - |\hat{v}|_j \ge |u|_j - |v|_j > C$. 

We now assume $|u| > |v|$. 
Take $\hat{v} = v$ and any factor $\hat{u}$ of~$u$ with $|\hat{u}| = |\hat{v}|$. 
One has $|\hat{u}|_j  + \big(|u| - |v|)   \ge |u|_j $.
We get $|\hat{u}|_j - |\hat{v}|_j \ge |u|_j - \big(|u| - |v|) - |v|_j > C$.
\end{proof}

\begin{lemma} \label{l:unbalanced2}
Let $u, v$ be factors of $\sigma_i(w)$ for some $w \in \mathcal{A}^*$ such that $|u| \ge |v|$ and
$$
|u|_j - |v|_j > C + \bigg\lceil\frac{|u|-|v|}{2}\bigg\rceil \quad \mbox{for some}\ j \in \mathcal{A} \setminus \{i\}\,.
$$
Then there exist factors $\hat{u}$ of~$u$ and $\hat{v}$ of~$v$ with $|\hat{u}| = |\hat{v}|$ and $|\hat{u}|_j - |\hat{v}|_j > C$.
\end{lemma}

\begin{proof}
Let $\hat{v} = v$ and $\hat{u}$ be the prefix (or suffix) of $u$ with $|\hat{u}| = |\hat{v}|$. 
Since there are no two consecutive letters $j$ in~$\sigma_i(w)$ and thus in~$u$, we obtain that $|u|_j \le |\hat{u}|_j + \big\lceil\frac{|u|-|v|}{2}\big\rceil$.
This implies $|\hat{u}|_j - |\hat{v}|_j \ge |u|_j - \big\lceil\frac{|u|-|v|}{2}\big\rceil - |v|_j > C$.
\end{proof}

\begin{proof}[Proof of Theorem~\ref{t:1}]
Let $\omega$ be an Arnoux-Rauzy word with $d = 3$ and weak partial quotients that are bounded by~$h$.
We will prove that not only~$\omega$, but all~$\omega^{(m)}$ are  $({2h\!+\!1})$-balanced. 
We will work by contradiction and consider a pair of factors of equal length that is not $({2h\!+\!1})$-balanced. 
We will furthermore take it with minimal length among all such pairs, and for all sequences~$\omega^{(m)}$.
This minimality assumption will prove to be crucial and provide the desired contradiction.

So, suppose that $\omega^{(m)}$ is not $({2h\!+\!1})$-balanced for some $m \ge 0$.
Let $u, v$ be a pair of factors of~$\omega^{(m)}$, $m \ge 0$, with $|u| = |v|$, that is not $({2h\!+\!1})$-balanced.
Assume w.l.o.g.\ that $|u|$ is minimal, that is, for any~$\ell \ge 0$, for any pair of factors $x,y$ of~$ \omega^{(\ell)}$ with $|x|=|y|< |u| =|v|$, and for any letter~$k$, one has $\big||x|_k - |y|_k\big| \leq 2h+1$.
This implies that $\big||u|_j - |v|_j\big|= 2h+2$  for some~$j \in \mathcal{A}$.
Moreover, we can assume w.l.o.g.\ that $m = 0$ and  that
$$
|u|_j - |v|_j = 2h+2\,. 
$$ 

We now want to reach a contradiction.

We can easily rule out the case $j = i_0$. 
Indeed, Lemma~\ref{l:preimage} gives $|u'|_{i_0} - |v'|_{i_0} = 2\, (2h+2) + \delta_u - \delta_v$ and $|u'| - |v'| = 2h+2 + \delta_u - \delta_v$, thus Lemma~\ref{l:unbalanced1} provides in turn factors $\hat{u}, \hat{v}$ of~$\omega^{(1)}$ with $|\hat{u}| = |\hat{v}| = |v'|$ that are not $({2h\!+\!1})$-balanced, with $|v'| < |v|$ (note that $|v| = |u| \ge |u|_j = |v|_j + 2h+2 \ge 2$), contradicting the minimality of~$|u|$. 

Assume from now on that $j \ne i_0$, thus $|u'|_j - |v'|_j = 2h+2$. 
If $|u'| \le |v'|$, then Lemma~\ref{l:unbalanced1} gives factors $\hat{u}, \hat{v}$ of~$\omega^{(1)}$ with $|\hat{u}| = |\hat{v}| =|v'|< |u|$ and $|\hat{u}|_j - |\hat{v}|_j \ge 2h + 2$, contradicting again the minimality of~$|u|$. 
Therefore, we must have $|u'|> |v'|$. 
By Lemma~\ref{l:preimage}, this implies in particular that $|u|_{i_0} - |v|_{i_0} \ge -1$. 
The equality $|u| = |v|$ then gives that $|u|_k - |v|_k \leq  -2h-1$ for the remaining third letter $k \in \mathcal{A} \setminus \{i_0, j\}$.

Observe that we also have $|u|_k - |v|_k \geq  -2h-1$. 
For if $|u|_k - |v|_k \le -2h-2$, then by Lemma~\ref{l:preimage} one has $|u'|_k - |v'|_k \le -2h-2$, and Lemma~\ref{l:unbalanced1}  together with the fact that  $|u'|> |v'|$ gives factors $\hat{u}, \hat{v}$ of~$\omega^{(1)}$ with $|\hat{u}| = |\hat{v}| =|v'|< |u|$ and $|\hat{v}|_k - |\hat{u}|_k \ge 2h+2$, contradicting the minimality of~$|u|$. 
Thus we obtain that
$$
|u|_{i_0} - |v|_{i_0} = -1\,, \quad |u|_j - |v|_j = 2h+2\,, \quad |u|_k - |v|_k = -2h-1\,,
$$
with $\mathcal{A} = \{i_0,j,k\}$.

Let $n \ge 1$ be such that $i_0 = i_1 = \cdots = i_{n-1} \ne i_n$.
Then
$$
|u^{(\ell)}|_j - |v^{(\ell)}|_j = 2h+2 \quad \mbox{and} \quad |u^{(\ell)}|_k- |v^{(\ell)}|_k = -2h-1 \quad \mbox{for}\ 0 \le \ell \le n.
$$
Therefore, the above arguments showing that $|u'| > |v'|$ now show that $|u^{(\ell+1)}| > |v^{(\ell+1)}|$ for $0 \le \ell < n$ (note that $|u^{(\ell)}| \le |u|$ by Lemma~\ref{l:preimage}).

For $1 \le \ell < n$, we do not have $|u^{(\ell)}|_{i_0} - |v^{(\ell)}|_{i_0} = -1$.
Nevertheless, we calculate
$$
|u^{(\ell+1)}| - |v^{(\ell+1)}| \le |u^{(\ell)}|_{i_0} - |v^{(\ell)}|_{i_0} + 2 = |u^{(\ell)}| - |v^{(\ell)}| + 1
$$
for $0 \le \ell < n$.
By the assumption of the theorem, we have $n \le h$, thus 
\begin{equation}\label{eq:h}
1 \le |u^{(n)}| - |v^{(n)}| \le n \le h\,.
\end{equation}

We cannot have $i_n = k$, since this would imply $|u^{(n+1)}| - |v^{(n+1)}| \le -2h+1 < 0$ and $|u^{(n+1)}|_j - |v^{(n+1)}|_j = 2h+2$, hence the existence of factors $\hat{u}, \hat{v}$ of~$\omega^{(n+1)}$ with $|\hat{u}| = |\hat{v}| < |u|$ and $|\hat{u}|_j - |\hat{v}|_j \ge 2h+2$ by Lemma~\ref{l:unbalanced1}.

Since $i_n \neq k$, we have $i_0 = i_1 = \cdots = i_{n-1} \ne i_n = j$.
Then, by Lemma~\ref{l:preimage},
$$
|u^{(n+1)}| - |v^{(n+1)}| \ge 2h\,.
$$ 
Let $r \ge 1$ be such that $i_n = i_{n+1} = \cdots = i_{n+r-1} \ne i_{n+r}$.
According to Lemma~\ref{l:preimage}, $|u^{(\ell)}|_k = |u^{(n)}|_k$, $|v^{(\ell)}|_k = |v^{(n)}|_k$, and similarly $|u^{(\ell)}|_{i_0} = |u^{(n)}|_{i_0}$, $|v^{(\ell)}|_{i_0} = |v^{(n)}|_{i_0}$,  for $n \le \ell \le n+r$.
We thus deduce from these equalities and from~(\ref{eq:h}) that 
$$
|u^{(\ell)}|_j - |v^{(\ell)}|_j - \big(|u^{(\ell)}| - |v^{(\ell)}|) = |u^{(n)}|_j - |v^{(n)}|_j - \big(|u^{(n)}| - |v^{(n)}|) \ge h+2
$$
for $n \le \ell \le n+r$, and then, from Lemma~\ref{l:preimage}, that 
$$
|u^{(\ell+1)}| - |v^{(\ell+1)}| \ge |u^{(\ell)}|_j - |v^{(\ell)}|_j - 2 \ge  |u^{(\ell)}| - |v^{(\ell)}| + h
$$
for $n \le \ell < n+r$, in particular $|u^{(n+r)}| - |v^{(n+r)}| \ge |u^{(n+1)}| - |v^{(n+1)}|$.
This implies
\begin{multline*}
|u^{(n+r)}|_j - |v^{(n+r)}|_j - \bigg\lceil\frac{|u^{(n+r)}|-|v^{(n+r)}|}{2}\bigg\rceil \\
= |u^{(n+r)}|_j - |v^{(n+r)}|_j - \big(|u^{(n+r)}|-|v^{(n+r)}|\big) + \bigg\lfloor\frac{|u^{(n+r)}|-|v^{(n+r)}|}{2}\bigg\rfloor \\
\ge h+2 + \bigg\lfloor\frac{|u^{(n+1)}|-|v^{(n+1)}|}{2}\bigg\rfloor \ge 2h+2\,.
\end{multline*}
Now, Lemma~\ref{l:unbalanced2} provides factors $\hat{u}, \hat{v}$ of~$\omega^{(n+r)}$ with $|\hat{u}| = |\hat{v}| = |v^{(n+r)}|$ and $|\hat{u}|_j - |\hat{v}|_j \ge 2h+2$.
Since $|v^{(n+r)}| <  |v|$ by Lemma~\ref{l:preimage}, we have obtained a 
contradiction to the minimality of~$|u|$ in the last remaining case as well.
Therefore, $\omega$~is $({2h\!+\!1)}$-balanced if the weak partial quotients are bounded by~$h$.
\end{proof}

\begin{proof}[Proof of Theorem~\ref{t:2}]
Similarly to the proof of Theorem~\ref{t:1}, suppose that $\omega^{(m)}$ is not $2$-balanced for some $m \ge 0$.
Let $u, v$ be factors of $\omega^{(m)}$, $m \ge 0$, with minimal length $|u| = |v|$ such that $|u|_j - |v|_j = 3$ for some $j \in \mathcal{A}$, and assume w.l.o.g.\ $m = 0$.

The proof of Theorem~\ref{t:1} shows that we must have 
$$
|u|_{i_0} - |v|_{i_0} = -1\,, \quad |u|_j - |v|_j = 3\,, \quad |u|_k - |v|_k = -2\,,
$$
with $\mathcal{A} = \{i_0, j, k\}$, and 
$$
|u^{(n)}|_j - |v^{(n)}|_j = 3\,, \quad |u^{(n)}|_k - |v^{(n)}|_k = -2\,, \quad 1 \le |u^{(n)}| - |v^{(n)}| \le n\,,
$$
for $n \ge 1$ satisfying $i_0 = i_1 = \cdots = i_{n-1} \ne i_n$.
Moreover, we must have $i_n = j$, thus
$$
|u^{(n+1)}|_k - |v^{(n+1)}|_k = -2\,, \quad |u^{(n+1)}| - |v^{(n+1)}| \ge |u^{(n)}|_j - |v^{(n)}|_j - 2 = 1\,.
$$

By assumption~(1) of the theorem, we have either $n \geq 2$ and $i_{n+1}=k$, or $n=1$.
  
If $i_{n+1} = k$ (which holds in particular if $n \ge 2$), then Lemma~\ref{l:preimage} gives 
\begin{align} 
& |v^{(n+2)}|_k - |u^{(n+2)}|_k - \big(|v^{(n+2)}| - |u^{(n+2)}|\big) \nonumber \\
& \qquad = |v^{(n+1)}|_k - |u^{(n+1)}|_k - \big(|v^{(n+1)}| - |u^{(n+1)}|\big) \ge 3 \label{e:dif}
\end{align}
and $|v^{(n+2)}|_k \ge |u^{(n+2)}|_k$.
Thus Lemma~\ref{l:unbalanced1} provides factors $\hat{u}, \hat{v}$ of~$\omega^{(n+2)}$ with $|\hat{u}| = |\hat{v}|= |u^{(n+2)}|$ and $|\hat{v}|_k - |\hat{u}|_k \ge 3$.
As $|u^{(n+2)}| < |u|$ by Lemma~\ref{l:preimage}, this contradicts the minimality of~$|u|$.

It remains to consider the case $n = 1$, $i_2 \ne k$. 
We have seen above that $n = 1$ means that $i_1 = j$, $|u^{(1)}|_j -  |v^{(1)}|_j = 3$, $|u^{(1)}|_k -  |v^{(1)}|_k = -2$, and $|u^{(1) }| - |v^{(1)}| = 1$.
Let $r \ge 1$ be such that $j = i_1 = i_2 = \cdots = i_r \ne i_{r+1}$.
Similarly to the proof of Theorem~\ref{t:1}, we get
$$
|u^{(\ell)}|_j -  |v^{(\ell)}|_j - \big(|u^{(\ell)}| - |v^{(\ell)}|\big) = |u^{(1)}|_j -  |v^{(1)}|_j - \big(|u^{(1)}| - |v^{(1)}|\big) = 2
$$ 
for $1 \le \ell \le r+1$, hence
$$
|u^{(\ell+1)}| - |v^{(\ell+1)}| \ge |u^{(\ell)}|_j - |v^{(\ell)}|_j - 2 = |u^{(\ell)}| - |v^{(\ell)}|
$$
for $1 \le \ell \le r$, thus 
$$
|u^{(r+1)}| - |v^{(r+1)}| \ge 1\,, \quad |u^{(r+1)}|_j - |v^{(r+1)}|_j \ge 3\,, \quad |u^{(r+1)}|_k - |v^{(r+1)}|_k = -2\,.
$$

We cannot have $i_{r+1} = k$. 
Indeed, we would have $|u^{(r+2)}|_j - |v^{(r+2)}|_j \ge 3$ and $|u^{(r+2)}| \le |v^{(r+2)}|$ by Lemma~\ref{l:preimage}, contradicting the minimality of~$|u|$ by Lemma~\ref{l:unbalanced1}.

We thus have $i_{r+1} = i_0$, hence $|u^{(r+2)}|_j - |v^{(r+2)}|_j \ge 3$ and $|u^{(r+2)}|_k - |v^{(r+2)}|_k = -2$ by Lemma~\ref{l:preimage}. 
Similarly to the preceding paragraph, we exclude $i_{r+2} = k$. 
Then assumption (1) of the theorem forces $r = 1$.
Let now $s \ge 1$ be such that $(i_0=)\, i_2 = i_3 = \cdots = i_{s+1} \ne i_{s+2}$.
Again, Lemma~\ref{l:preimage} gives $|u^{(s+2)}|_j - |v^{(s+2)}|_j \ge 3$ and $|u^{(s+2)}|_k - |v^{(s+2)}|_k = -2$, and we can exclude $i_{s+2} = k$, i.e., we have $i_{s+2} = j$.
By Lemma~\ref{l:preimage}, we have
$$
|u^{(s+3)}| - |v^{(s+3)}| \ge |u^{(s+2)}|_j - |v^{(s+2)}|_j - 2 \ge 1\,, \quad |u^{(s+3)}|_k - |v^{(s+3)}|_k = -2\,.
$$

Since $i_0 = i_s = i_{s+1} \ne i_{s+2} = j$ in case $s \ge 2$, and $i_0 = i_2 \ne i_3 = i_1 = j$ in case $s = 1$, the assumptions of the theorem force now $i_{s+3} = k$.
Replacing $n$ by~${s\!+\!2}$ in~(\ref{e:dif}) and the subsequent lines, we get factors $\hat{u}, \hat{v}$ of~$\omega^{(s+4)}$ with $|\hat{u}| = |\hat{v}| = |u^{(s+4)}| < |u|$ and $|\hat{v}|_k - |\hat{u}|_k \ge 3$, providing the desired contradiction to the minimality of~$|u|$.
\end{proof}

\subsection{Imbalances} \label{subsec:unbalance}

In order to prove Theorem~\ref{t:3} and Proposition~\ref{prop:2}, we first need a lemma that allows us to exhibit pairs of factors in~$\omega^{(m)}$ having prescribed differences of Parikh vectors, with the \emph{Parikh vector}  associated with the word~$w$ being defined as $(|w|_1,|w|_2,|w|_3)$.

\begin{lemma} \label{l:subst}
Let $u, v$ be factors of $\omega^{(m+1)}$, $m \ge 0$. 
Then, for any $\delta \in \{0, \pm1, \pm2\}$, there exist factors $\hat{u}, \hat{v}$ of~$\omega^{(m)}$ such that
$$
|\hat{u}|_{i_m} - |\hat{v}|_{i_m} = |u| - |v| + \delta\,, \quad |\hat{u}|_j - |\hat{v}|_j = |u|_j - |v|_j\ \mbox{for all}\ j \in \mathcal{A} \setminus \{i_m\}\,.
$$
\end{lemma}

\begin{proof}
Clearly, $\sigma_{i_m}(u)$ and $\sigma_{i_m}(v)$ are factors of~$\omega^{(m)}$, with
$$
|\sigma_{i_m}(u)|_{i_m} - |\sigma_{i_m}(v)|_{i_m} = |u| - |v|, \ |\sigma_{i_m}(u)|_j - |\sigma_{i_m}(v)|_j = |u|_j - |v|_j\ \mbox{for}\ j \in \mathcal{A} \setminus \{i_m\}\,.
$$
Since $\sigma_{i_m}(u)$ and $\sigma_{i_m}(v)$ end with~$i_m$, and $i_m\, \sigma_{i_m}(u)$ and $i_m\, \sigma_{i_m}(v)$ are factors of~$\omega^{(m)}$, we obtain every $\delta \in \{\pm1, \pm2\}$ by removing $i_m$ at the end of $\sigma_{i_m}(u)$ and/or $\sigma_{i_m}(v)$ and/or appending $i_m$ in front of $\sigma_{i_m}(u)$ and/or $\sigma_{i_m}(v)$.
\end{proof}

\begin{proof}[Proof of Theorem~\ref{t:3}]
We first prove that $\omega^{(m)}$ contains factors $u, v$ whose difference of Parikh vectors $(|u|_1 - |v|_1, |u|_2 - |v|_2, |u|_3 - |v|_3)$ is $(1, 2, -2)$ or $(-1, 2, -2)$. 
Starting from such a pair $(u,v)$ and applying the substitution corresponding to any given word in $\{1,3\}^*\, 1\, 2^*\, 1\, 2\, \{1,2\}$, we will then exhibit an imbalance of~$3$ in~$\omega$.

By Theorem~\ref{theo:ar}, the sequence $(i_m)_{m\ge0}$ contains infinitely many occurrences of~$1$, $2$, and $3$ respectively.
Therefore, for any $m \ge 0$, $i_m i_{m+1} \cdots$ starts with a word in $1^*\, 2\, \mathcal{A}^*\, 2\, \mathcal{A}^*\, 3\, \mathcal{A}^*\, 3$ or $1^*\, 3\, \mathcal{A}^*\, 3\, \mathcal{A}^*\, 2\, \mathcal{A}^*\, 2$.

Assume first $i_m \cdots i_{m+3} = 2 2 3 3$.
Starting with the factor $u = 1$ of~$\omega^{(m+4)}$ and $v$ the empty word, and successively applying Lemma~\ref{l:subst}, we obtain the following  chain of vectors $(|\hat{u}|_1 - |\hat{v}|_1, |\hat{u}|_2 - |\hat{v}|_2, |\hat{u}|_3 - |\hat{v}|_3)$ for factors $\hat{u}, \hat{v}$ of~$\omega^{(n)}$, $m+4 \ge n \ge m$:
$$
\begin{pmatrix}1\\0\\0\end{pmatrix} \stackrel{\displaystyle\sigma_3}\longrightarrow \begin{pmatrix}1\\0\\-1\end{pmatrix} \stackrel{\displaystyle\sigma_3}\longrightarrow \begin{pmatrix}1\\0\\-2\end{pmatrix} \stackrel{\displaystyle\sigma_2}\longrightarrow \begin{pmatrix}1\\1\\-2\end{pmatrix} \stackrel{\displaystyle\sigma_2}\longrightarrow \begin{pmatrix}1\\2\\-2\end{pmatrix}.
$$

We can insert between the terms of this chain some substitutions that will not alter the difference vectors that have been created.  
Indeed, except for the last vector, any of these vectors is the image of itself by any substitution~$\sigma_j$, $j \in \mathcal{A}$, and some $\delta \in \{0, \pm1, \pm 2\}$ as in Lemma~\ref{l:subst}. 
More precisely, if $(u,v) $ are factors of~$\omega^{(n+1)}$ satisfying $(|u|_1 - |v|_1, |u|_2 - |v|_2, |u|_3 - |v|_3) = (1,0,-1)$, then whatever the value of~$i_n$, there exist factors $\hat{u}, \hat{v}$ of~$\omega^{(n)}$ with the same difference vector $(1,0,-1)$.
The same holds for the vectors $(1,0,-2)$ and $(1,1,-2)$, while the last vector $(1,2,-2)$ is stabilized only by $\sigma_1$ and~$\sigma_2$. 
Therefore, if $i_m i_{m+1} \cdots$ starts with a word in $1^*\, 2\, \mathcal{A}^*\, 2\, \mathcal{A}^*\, 3\, \mathcal{A}^*\, 3$, then $\omega^{(m)}$ contains factors $u, v$ with $(|u|_1 - |v|_1, |u|_2 - |v|_2, |u|_3 - |v|_3) = (1, 2, -2)$.

Symmetrically, if $i_m i_{m+1} \cdots$ starts with a word in $1^*\, 3\, \mathcal{A}^*\, 3\, \mathcal{A}^*\, 2\, \mathcal{A}^*\, 2$, then the chain
$$
\begin{pmatrix}-1\\0\\0\end{pmatrix} \stackrel{\displaystyle\sigma_2}\longrightarrow \begin{pmatrix}-1\\1\\0\end{pmatrix} \stackrel{\displaystyle\sigma_2}\longrightarrow \begin{pmatrix}-1\\2\\0\end{pmatrix} \stackrel{\displaystyle\sigma_3}\longrightarrow \begin{pmatrix}-1\\2\\-1\end{pmatrix} \stackrel{\displaystyle\sigma_3}\longrightarrow \begin{pmatrix}-1\\2\\-2\end{pmatrix}
$$
gives factors $u, v$ of~$\omega^{(m)}$ with $(|u|_1 - |v|_1, |u|_2 - |v|_2, |u|_3 - |v|_3) = (-1, 2, -2)$. This finishes our first step.

Starting with factors $u, v \in \omega^{(m)}$ such that $(|u|_1 - |v|_1, |u|_2 - |v|_2, |u|_3 - |v|_3) = (\pm1, 2, -2)$, we now use the chain
$$
\begin{pmatrix}\pm1\\2\\-2\end{pmatrix} \stackrel{\displaystyle\sigma_1}\longrightarrow \begin{pmatrix}1\\2\\-2\end{pmatrix} \stackrel{\displaystyle\sigma_2}\longrightarrow \begin{pmatrix}1\\3\\-2\end{pmatrix} \stackrel{\displaystyle\sigma_1}\longrightarrow \begin{pmatrix}0\\3\\-2\end{pmatrix} \stackrel{\displaystyle\sigma_1}\longrightarrow \begin{pmatrix}-1\\3\\-2\end{pmatrix}.
$$
Here, $(0, 3, -2)$ is fixed by $\sigma_2$ (with $\delta = 2$); $(-1, 3, -2)$ is fixed by $\sigma_1$ (with $\delta = -1$) and by $\sigma_3$ (with $\delta = -2$), i.e., this chain can be used for any word in $\{1,3\}^*\, 1\, 2^*\, 1\, 2\, 1$.
When $(|u|_1 - |v|_1, |u|_2 - |v|_2, |u|_3 - |v|_3) = (1, 2, -2)$, then the first transition is of course not needed and we can use the rest of the chain for any word in $\{1,3\}^*\, 1\, 2^*\, 1\, 2$.

We have thus shown that, if $i_0 i_1 i_2 \cdots$ starts with a word in $\{1,3\}^*\, 1\, 2^*\, 1\, 2\, 1$, then $\omega$ contains factors $u, v$ such that $(|u|_1 - |v|_1, |u|_2 - |v|_2, |u|_3 - |v|_3) = (1,3,-2)$, which implies that $\omega$ is not $2$-balanced.
If a prefix of $i_0 i_1 i_2 \cdots$ is in $\{1,3\}^*\, 1\, 2^*\, 1\, 2\, 2$, then we get the same result, as $i_0 i_1 i_2 \cdots$ starts now with a word in $\{1,3\}^*\, 1\, 2^*\, 1\, 2$ followed by an element of $2\, \mathcal{A}^*\, 2\, \mathcal{A}^*\, 3\, \mathcal{A}^*\, 3$.
\end{proof}

\begin{proof}[Proof of Proposition~\ref{prop:2}]
Let $\omega$ satisfy the conditions of the proposition.
Since the conditions of Theorem~\ref{t:2} are violated only for $m = 2$, the Arnoux-Rauzy words~$\omega^{(m)}$, $m \ge 2$, satisfy all conditions of Theorem~\ref{t:2} and are thus $2$-balanced.
By Theorem~\ref{t:3}, $\omega^{(1)}$~is not $2$-balanced.

The proof of the $2$-balancedness of~$\omega$ will be similar to the proofs of Theorems~\ref{t:1} and~\ref{t:2}: 
Supposing that $\omega$ contains a pair of factors that are $2$-unbalanced, we construct from this pair a $2$-unbalanced pair of factors of~$\omega^{(m)}$.
The only difference is that we require now $m \ge 2$. 

We first show that $\omega^{(1)}$ has no factors $u,v$ with $|u| = |v|$ and $|u|_1 - |v|_1 = 3$ or $|u|_3 - |v|_3= 3$.
To this end, recall that the $\mathcal{S}$-directive sequence of~$\omega^{(1)}$ starts with~$112$. 
The proof of Theorem~\ref{t:1} (see also the proof of Theorem~\ref{t:2}) shows that if $\omega^{(1)}$ had factors~$u,v$ with $|u| = |v|$ and $|u|_1 - |v|_1 = 3$, then the same would be true for~$\omega^{(2)}$, contradicting that $\omega^{(2)}$ is $2$-balanced.
We also see from the proof of Theorem~\ref{t:1} that the existence of factors~$u,v$ of~$\omega^{(1)}$ with $|u| = |v|$ and $|u|_3 - |v|_3 = 3$ contradicts that $\omega^{(m)}$ is $2$-balanced for all $m \in \{2,3,4\}$. 

Now suppose that $\omega$ is not $2$-balanced, i.e., there exist factors $u, v$ of~$\omega$ with $|u| = |v|$ and $|u|_j - |v|_j = 3$ for some $j \in \mathcal{A}$. 
Recall that $i_0 i_1 \cdots i_5 = 211213$.
 
Assume first $j = 3$. 
By Lemma~\ref{l:preimage}, we have $|u^{(2)}|_3 - |v^{(2)}|_3 = |u^{(1)}|_3 - |v^{(1)}|_3 = 3$, $|u^{(1)}| - |v^{(1)}| \le |u|_2 - |v|_ 2 + 2$, and $|u^{(2)}| - |v^{(2)}| \le |u^{(1)}|_1 - |v^{(1)}|_ 1 + 2 = |u|_1 - |v|_ 1 + 2$.
Since $|u| = |v|$ and $|u|_3 - |v|_3 = 3$, we have either $|u|_1 - |v|_ 1 \le -2$ or $|u|_2 - |v|_ 2 \le -2$, thus $|u^{(2)}| \le |v^{(2)}|$ or $|u^{(1)}| \le |v^{(1)}|$.
By Lemma~\ref{l:unbalanced1}, this contradicts that $\omega^{(2)}$ is $2$-balanced or that imbalances of size~$3$ do not occur in~$\omega^{(1)}$ for the letter~$3$.

Assume now $j = 2$. 
We deduce from Lemma~\ref{l:preimage} that $|u^{(1)}|_2 - |v^{(1)}|_2 \ge 4$, $|u^{(1)}| - |v^{(1)}| \ge 1$. 
In the last paragraph we have seen that $|u|_3 - |v|_3 \le -3$ is impossible.
This implies $|u|_1 - |v|_1 \le -1$, hence $|u^{(1)}|_1 - |v^{(1)}|_1 \le -1$ by Lemma~\ref{l:preimage}.
Now we obtain by Lemma~\ref{l:preimage} that $|u^{(2)}| - |v^{(2)}| \le 1$ and $|u^{(2)}|_2 - |v^{(2)}|_2 \ge 4$, which by Lemma~\ref{l:unbalanced1} contradicts the $2$-balancedness of~$\omega^{(2)}$.

Finally, assume $j = 1$. 
Then Lemma~\ref{l:preimage} provides $|u^{(1)}|_1 - |v^{(1)}|_1  =3$.
We cannot have $|u^{(1)}| \le |v^{(1)}|$, as this contradicts by Lemma~\ref{l:unbalanced1} that imbalances of size~$3$ do not occur in~$\omega^{(1)}$ for the letter~$1$.
Therefore, we must have $|u|_2 - |v|_ 2 \ge -1$ by Lemma~\ref{l:preimage}.
As $|u|_3 - |v|_3 \le -3$ is impossible, we get 
$$
|u|_1 - |v|_ 1  = 3\,, \qquad |u|_2 - |v|_ 2  = -1\,, \qquad |u|_3 - |v|_ 3  = -2\,,
$$
Using that $|u^{(1)}| > |v^{(1)}|$, Lemma~\ref{l:preimage} gives 
$$
|u^{(1)}|_1 - |v^{(1)}|_ 1  = 3\,, \qquad |u^{(1)}|_2 - |v^{(1)}|_ 2 = 0\,, \qquad |u^{(1)}|_3 - |v^{(1)}|_ 3  = -2\,.
$$
As $i_1 = i_2 = 1$, $i_3 = 2$, further applications of Lemma~\ref{l:preimage} give
\begin{align*}
& |u^{(4)}|_1 -  |v^{(4)}|_1 = |u^{(3)}|_1 -  |v^{(3)}|_1 \ge 2\, \big(|u^{(2)}|_1 -  |v^{(2)}|_1\big) - \big(|u^{(2)}| - |v^{(2)}|\big) - 2 \\ 
& = |u^{(2)}|_1 -  |v^{(2)}|_1 + |u^{(1)}|_1-  |v^{(1)}|_1 - \big(|u^{(1)}| - |v^{(1)}|\big) - 2 = |u^{(2)}|_1 -  |v^{(2)}|_1 \ge 3
\end{align*}
and, with $i_4 = 1$, 
$$
|u^{(5)}| -  |v^{(5)}| \ge 1\,, \qquad |u^{(5)}|_3 -  |v^{(5)}|_3 = |u|_3 -  |v|_3 = -2\,.
$$
Since $i_5 = 3$, we can replace $n$ by~$4$ and $k$ by~$3$ in~(\ref{e:dif})  and the subsequent lines, and get a contradiction to the $2$-balancedness of~$\omega^{(6)}$.
\end{proof}

\section*{Acknowledgments}
We warmly thank S\'ebastien Ferenczi and Pierre Arnoux for useful discussions on the subject.

\end{document}